\documentclass[superscriptaddress,pra,twocolumn,floatfix]{revtex4-2}
\usepackage{amsmath,amssymb,amsthm}
\usepackage{epsfig}
\usepackage{physics}
\usepackage{graphicx}
\usepackage{dcolumn}
\usepackage{bm}
\usepackage{url}
\usepackage{color}
\usepackage{comment}
\usepackage{units}
\usepackage{enumerate}
\usepackage{xr}
\usepackage[hyperindex,breaklinks]{hyperref}
\usepackage{mathtools}
\usepackage{comment}
\usepackage{dsfont}
\usepackage{mathrsfs}
\usepackage{braket}
\usepackage{algorithm}
\usepackage{algpseudocode}
\usepackage{tikz}
\usetikzlibrary{quantikz}
\usetikzlibrary{positioning}
\newtheorem{lemma}{Lemma}
\theoremstyle{remark}
\newtheorem{fact}{Fact}

\theoremstyle{definition}

\hypersetup{%
	breaklinks=true,   
	colorlinks=true,   
	pdfusetitle=true,  
	allcolors=blue
}
\newcommand\bigcircle[1][2]{{\mathop{\tikz[baseline=-2.8]{\draw[thick](0,0)circle[radius=#1mm];}}}}

\algrenewcommand\algorithmicrequire{\textbf{Input:}}
\algrenewcommand\algorithmicensure{\textbf{Output:}}

\begin{document}
\title{Blindly Factorizing 21 Quantumly}
\author{Aritra Das}
\email{aritra.das@anu.edu.au}
\affiliation{%
Institute for Quantum Science and Technology, University of Calgary, Alberta T2N~1N4, Canada.}
\affiliation{%
	Centre for Quantum Computation and Communication Technology, Department of Quantum Science, Australian National University, ACT 2601, Australia.}
\author{Barry~C. Sanders}
\email{sandersb@ucalgary.ca}
\affiliation{%
Institute for Quantum Science and Technology, University of Calgary, Alberta T2N~1N4, Canada.}
\date{\today}
\begin{abstract}
We develop a classically verifiable scheme
for blindly factorizing the semiprime 21 quantumly
for a classical client who does not trust the remote quantum servers.
Our scheme advances state of the art,
which achieves blind factorization of 15 quantumly,
by increasing the problem to factorizing the next semiprime,
choosing a harder base, executing a non-Clifford gate, and
showing that the security check for 15 also works for 21.
Our algorithmic approach to incorporating non-Clifford operations
sets the stage for scaling blind quantum factorization, whereas
our five-EPR-pair scheme motivates a photonic experiment
that supplants current demonstrations of blind factorization.
\end{abstract}
\maketitle
\section{Introduction}
\label{sec:Intro}
Nowadays commercial quantum computers,
which strive for near-term intermediate-scale quantum advantages~\cite{PreskillNISQ,PreskillSupremacy},
are accessed by the cloud~\cite{Devitt}
raising the problem of a client
Clara (C)
not trusting the remote server(s).
Risk mitigation strategies include blind quantum computing (BQC)
for weakly quantum clients
(using either single-server prepare-and-send
or receive-and-measure protocols)~\cite{Childs, BFK09, Fitzsimons, FK17,BKB+12, GKK19}
and for purely classical clients
(using multi-server entanglement-based protocols)~\cite{GKK19, HZM+17,CCKW21} and quantum homomorphic encryption~\cite{BJ15,OTF18,TFB+20,OTFR20}
whereby~C delegates quantum computation to one or more remote servers, who
are denied key information about the computation~\cite{Fitzsimons}.
Building on successful experimental factorization
of the odd semiprime
(odd integer~$N=pq$ for~$p,q\in\mathbb{P}$ and~$p\neq q$)
$N=15$~\cite{HZM+17},
we devise a protocol for C
to delegate secure factorization of $N=21$~\cite{Martin}
to two remote quantum servers
called Alice (A) \&\ Bob (B).

Our approach extends the
BQC factorization of~15
in two ways~\cite{HZM+17}.
First we increase~$N$ from~15 to
the next odd semiprime number~21.
Second,
we choose a harder base $a$~\cite{SSV13}
for modular exponentiation (modexp) $f(x):=a^x \bmod N$
(where $\gcd(a,N)=1$ with~$\gcd$ denoting greatest common divisor).
The period~$r$ of~$f(x)$
yields a solution~$p=\gcd(a^{\nicefrac{r}2}+1, N)$ when
the following two conditions are simultaneously met:
(i) either~$r$ is even, or~$r$ is odd and~$a$ is a perfect square,
and (ii)~$a^{\nicefrac{r}2}\not\equiv-1\bmod{N}$.
Period finding is sped up
subexponentially by quantum computing~\cite{Shors}.
For~$N=15$ with~$a=11$,
$r=2$ was achieved experimentally~\cite{HZM+17};
in contrast, we treat the hard case~$N=21$ with~$a=4$, for which~$r=3$.
This harder~$a$ requires incorporating a non-Clifford operator,
for which we employ the controlled-controlled-not (C$^2$NOT or Toffoli) gate~\cite{Shi03}.

The remainder of our paper is orgainsed as follows:
In Sec.~\ref{sec:Background}, we summarily recall
some relevant prerequisites to our work
and point to comprehensive resources on key topics.
In Sec.~\ref{sec:Approach}, we describe our methodology to construct
a blind quantum factorization scheme for given~$N$ and~$a$.
In Sec.~\ref{sec:Results} we present a formal algorithm,
along with a function library, to design blind quantum factorization
circuits for arbitrary~$N$ and~$a$
and also present the resulting circuits for two cases of~$N=21, a=4$.
We finish with a discussion on the significance of our results
in Sec.~\ref{sec:Disc} and a conclusion in Sec.~\ref{sec:Conc}.

\section{Background}
\label{sec:Background}

In this section we briefly review
the state-of-the-art in BQC,
blind quantum factorization,
and some other key concepts
that are fundamental to our work.
BQC is a quantum cryptographic protocol that
allows clients with limited or no quantum hardware
to outsource a comptuation to remote quantum server(s)
without revealing information about the computation itself to the server(s)~\cite{BFK09}.
Several BQC protocols have already been developed
and demonstrated for weakly quantum clients~\cite{Childs, BFK09, Fitzsimons, FK17},
but a purely classical client communicating only classically
with a single quantum server might not be able to achieve secure BQC~\cite{Aaronson}.
Nevertheless, this obstacle is overcome
if multiple servers sharing non-local resources are employed~\cite{RUV}.
A brief overview of verifiable BQC can be found in Ref.~\cite{GKK19}.

Secure BQC for completely classical clients, thus,
warrants the remote and classical leveraging of
quantum-advantageous algorithms,
like Shor's factorization~\cite{Shors}
or Grover's search~\cite{Grover},
which serve as prime candidates for delegation~\cite{BKB+12, HZM+17}.
Delegated Shor's factorization is known to be feasible
in the measurement-based quantum computation model~\cite{BFK09}
and has been demonstrated experimentally for $N=15$
in the quantum circuit model~\cite{HZM+17}.
The approach in both these works comprises C
delegating the quantum period-finding subroutine,
which computes the period~$r>1$ of~$f(x)$
for a given odd semiprime~$N=pq$ with unknown~$p$ \&~$q$,
of Shor's algorithm to remote server(s).
From the outputs returned by the servers to her,
C classically computes the factors as
\begin{equation}
	p,q=\gcd(a^{\nicefrac{r}2}\pm 1, N).
\end{equation}

A proof-of-principle implementation of BQC for a completely classical client
has been demonstrated in Ref.~\cite{HZM+17},
wherein Shor's algorithm~\cite{Shors} is executed for factorizing
$N=15$ using verifiable BQC based on the RUV protocol~\cite{RUV}.
This blind quantum factorization was performed for the choice of base~$a=11$,
which results in~$r=2$, thus
making the experimental demonstration sufficiently challenging for a proof-of-concept
but not as realistic as, for example, the case~$a=7$ and~$r=4$ would be.
This is because~$r=2$ implies that the quantum period-finding circuit
has reduced to a classical coin-toss experiment~\cite{SM09}---an anomaly
that can be rationalised as the choice of base~$a=11$
(implicitly) assuming pre-knowledge of the factors~\cite{SM09}.

A pre-knowledgeless factorization scheme would have to
choose a random base~$a$ from some set of allowed bases.
Without any prior ansatz,
such a choice would yield a hard base with high probability;
a hard base implies a period~$r>2$
and a period-finding circuit requiring the multi-qubit Toffoli gate,
which is a non-Clifford operator.
Introducing
a non-Clifford operator
brings
the quantum resource called ``magic''
into play~\cite{MagicState}.
``Magic'' enables quantum circuits to violate conditions for efficient classical simulatability~\cite{Gottesman1, Gottesman2}
so its inclusion is important for
scaling considerations concerning BQC factorization's quantum advantage.
To this end,
we now succinctly summarize the Clifford hierarchy of unitary operators.

The~$n$-qubit Pauli group is
\begin{equation}
\bm{C}^{(1)}_n := \{\pm1, \pm\text{i}\}\times \{ I, X, Y, Z \}^{\otimes n},
\end{equation}
qubits being two-level systems spanned by logical states~$\ket0,\ket1\in\mathscr{H}_2$,
and~$\mathscr{H}_d$ a $d$-dimensional Hilbert space.
Logical states are~$Z$-eigenstates
and comprise our computational basis,
and
\begin{equation}
X,Y,Z\in\mathcal{U}(\mathscr{H}_2)
\end{equation}
are the single-qubit Pauli operators,
with~$\mathcal{U}(\mathscr{H}_d)$ the group of unitary operators on~$\mathscr{H}_d$.
The $n$-qubit Clifford group $\bm{C}^{(2)}_n$
is the normalizer
of the Pauli group, i.e.,
\begin{equation}
\bm{C}^{(2)}_n:=\{ u\in \mathcal{U}(\mathscr{H}_{2^n}) ;  u \,\bm{C}^{(1)}_n u^\dagger\subseteq\bm{C}^{(1)}_n \}.
\end{equation}
This group is generated by the Hadamard, phase (phase-shift by~$\pi/2$) and controlled-not (CNOT) gates.

The Pauli and Clifford groups constitute
the first two levels of the Clifford hierarchy,
and the subsequent levels~$\bm{C}^{(k>2)}_n$  are
defined recursively by~\cite{Gottesman3}
\begin{equation}
\bm{C}_n^{(k)} := \{ u\in \mathcal{U}(\mathscr{H}_{2^n}); u\,\bm{C}^{(1)}_n u^\dagger
\subseteq\bm{C}_n^{(k-1)}\}.
\end{equation}
Conjugation with~$\bm{C}^{(2)}_n$ maps~$\bm{C}^{(1)}_n$ into itself,
so Clifford operators can be blindly delegated
using one-time Pauli pads~\cite{Childs}.
However, $\bm{C}^{(2)}_n$ does not constitute a universal gate-set.
Moreover, stabilizer quantum circuits,
which comprise only Clifford operators
and computational basis measurements
(corresponding to a projective-valued measure
$\ket{\epsilon}\bra{\epsilon}$
for $\ket\epsilon$
a computational basis state)~\cite{Clifford2},
can be simulated
efficiently (polynomial-time)
classically~\cite{Gottesman1, Gottesman2}.
The experimental blind quantum factorization of~15
requires only a stabilizer circuit in its simplest form~\cite{HZM+17}.

In contrast, the C$^2$NOT gate
along with~$\bm{C}^{(2)}_n$
constitutes a universal gate-set~\cite{BMPRV00,Unweyling}.
A circuit comprising both Clifford and C$^2$NOT gates
also circumvents the simulatability theorem~\cite{Gottesman1, Gottesman2}.
However, the inclusion of ``magic''
entails significant resource costs~\cite{Shi03, BJ15}.
Importantly, for our protocol,
only one of A \&\ B needs ``magic''
whereas the other executes a stabilizer circuit,
thereby simplifying the scheme for experimental realization.

\section{Approach}
\label{sec:Approach}

In this section we explain our approach
to solving the blind quantum factorization of~21.
First we describe the setup for blind quantum factorization,
namely the classical client, the bipartite quantum server
and their collective resources.
Next we describe our mathematical representation
of the computation circuit~$\mathcal{C}$ to be delegated to the servers
and~$\mathcal{C}$'s associated representations.
Our scheme for blind quantum factorization
relies upon computation by teleportation
on maximally entangled states, as identified in~\cite{Gottesman3};
in spirit, this is similar to the RUV protocol~\cite{RUV}
but we highlight some important distinctions in Sec.~\ref{sec:Disc}.
Next we establish the mathematical backbone
of our scheme---a procedure to obtain two blind circuits, one for each server,
from the circuit the classical client wishes to execute---via
Lemma~\ref{lemma:concat=simult} and Fact~\ref{fact:transpose=reverse}.
We conclude this section
by reviewing the full procedure to obtain the blind quantum circuits
from the input quantum circuit.

Similar to the BQC factorization of~15,
which we summarize in Fig.~\ref{fig:schemefor15}~\cite{HZM+17},
our scheme is based on the
Reichardt-Unger-Vazirani (RUV) protocol~\cite{RUV}.
The RUV protocol is a multi-round, two-server BQC scheme for a classical client
(as single-server BQC is not secure for classical clients~\cite{Aaronson, Morimae}).
In each round of our protocol, servers~A \&~B
receive~$n$ copies of the entangled two-qubit pair
$\ket{\Phi}:=\ket{00}+\ket{11}\in\mathscr{H}_{4}$
(for~$\ket{00}\equiv\ket0\otimes\ket0$,
and implied normalization employed throughout)
from a periodic source of entanglement, Deborah (D).
Each server receives one qubit from each copy of~$\ket\Phi$
and, thus, A \&\ B collectively share the resource
\begin{equation}
\label{eq:sharedresource}
\ket{\Phi}^{\otimes n}
=(\mathds1\otimes\mathds1 + X\otimes X)^{\otimes n}
\ket{00}^{\otimes n}
\in\mathscr{H}_{4^n}
\end{equation}
but have no other means for communicating~\cite{RUV}.
We index the~$2n$ qubits in~$\ket\Phi^{\otimes n}$
as shown in Fig.~\ref{fig:schemefor15} for~$n=3$.

To delegate an~$n$-qubit quantum circuit~$\mathcal{C}$
in the RUV protocol,
C instructs each server to either compute,
by executing a quantum circuit~($\mathcal{A}$ for A,~$\mathcal{B}$ for B),
or perform the measurement part of a Clauser-Horne-Shimony-Holt (CHSH) test~\cite{CHSH}.
There are, thus, four distinct subprotocols:
A~\&~B could both compute (computational subprotocol),
or both measure (CHSH subprotocol),
or else one computes while the other measures
(two tomography subprotocols)~\cite{RUV}.
When both compute,
A~\&~B report to C their~$i^\text{th}$ $Z$-measurement outcomes
\begin{equation}
\{(a_i,b_i);a_i,b_i\in\{0,1\},i\in[n]:=\{1, \dots, n\}\}.
\end{equation}
From their combined outcomes, C recovers the output from~$\mathcal{C}$
whereas the output from either~$\mathcal{A}$ or~$\mathcal{B}$
alone yields no information about~$\mathcal{C}$ except depth,
thereby blinding~A~\&~B.

\begin{figure}
\begin{tikzpicture}
	\node [anchor=south west,inner sep=0] (image) at (0,0) {\includegraphics[width=6 cm]{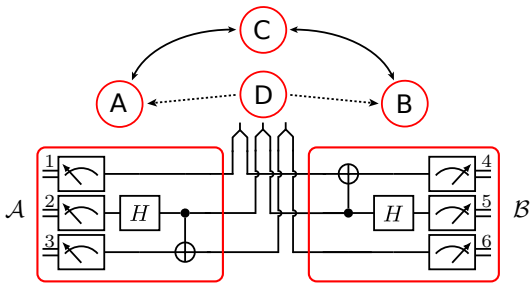}};
	\begin{scope}[x={(image.south east)},y={(image.north west)}]
		\draw [thick, red,rounded corners=3] (-0.01, -0.03) rectangle (0.4, 0.47);
		\draw [thick, red,rounded corners=3] (0.59, -0.03) rectangle (1.01, 0.47);
		\node [black] at (-0.06, 0.24) {$\mathcal{A}$};
		\node [black] at (1.06, 0.24) {$\mathcal{B}$};
		\node [black] at (0.0173, 0.419) {\scriptsize 1};
		\node [black] at (0.017, 0.27) {\scriptsize 2};
		\node [black] at (0.017, 0.122) {\scriptsize 3};
		\node [black] at (0.983, 0.419) {\scriptsize 4};
		\node [black] at (0.983, 0.27) {\scriptsize 5};
		\node [black] at (0.983, 0.122) {\scriptsize 6};
	\end{scope}%
\end{tikzpicture}%
\caption{%
	BQC scheme for factorizing~15.
	Quantum servers A \&\ B
	jointly compute circuits~$\mathcal{A}$ \&~$\mathcal{B}$
	(rounded rectangles), respectively,
	on state~$\ket{\Phi}^{\otimes 3}$
	supplied by entanglement source D,
	and report outcomes to classical client C.
	Each of $\mathcal{A}$ \&~$\mathcal{B}$ involves
	a CNOT and a Hadamard gate ($H$), and~$Z$-basis measurements.
	Solid and dotted arrows represent
	classical and quantum communication,
	respectively,
	arrowheads indicate directionality,
	and numbers represent indices for qubits.}%
\label{fig:schemefor15}%
\end{figure}

Now we discuss the underlying primitive gates
for~$\mathcal{C}$ over~$n$ qubits.
Each computational cycle,
with execution by circuit component~$\mathcal{C}_\nu$,
allows one or more of the following primitive gates
operating in parallel:
\begin{enumerate}[(i)]
\item single-qubit Hadamard~($H$),
\item single-qubit NOT~($X$),
\item two-qubit controlled-rotation~(CR$^k$),
and
\item multi-controlled NOT~(C$^l$NOT,
where~$l\in[n-1]$ denotes the number of controls),
also known as the multi-controlled Toffoli gate.
\end{enumerate}
Our~CR$^k$ gates are restricted to rotations
$\mathrm{R}^k \textcolor{red}{:} \, 0\leq k<n$
that impart a phase of~$\nicefrac{\pi}{2^k}$ on~$\ket1$
and zero phase on~$\ket0$.
Our choice of primitives is natural for Shor factorization~\cite{Shors},
and these primitives are composites of the ``standard set''
of CNOT,~$H$ and the~R$^2$ gate called~$T$,
as described in the Appendix A~\cite{BMPRV00,BBCD+95,MDM05,GKMR14,SM09}.

In Appendix B, we show that
there are at most
\begin{equation}
\left ( \frac{1133233}{8!}  \right) \, n^{\nicefrac{n}{2}} \, n!
\end{equation}
allowed circuit components over~$n$ qubits,
so
we can label~$\mathcal{C}_\nu$ by a bit string~$\bm{B}(\mathcal{C}_\nu)$
with size at most
\begin{equation}
\left\lceil\log \{(\nicefrac{1133233}{8!}) n^{\nicefrac{n}{2}} n! \}\right\rceil \, .
\end{equation}
The depth~$d$ circuit~$\mathcal{C}$
is then a composition of~$d$ circuit components,
\begin{equation}%
\label{eq:cyclecomposition}%
\mathcal{C}%
=\bigcircle_{\nu=1}^d \mathcal{C}_{\nu}%
:= \mathcal{C}_d \circ \cdots \circ \mathcal{C}_1  ,%
\end{equation}%
where~$\mathcal{C}_\nu$ is the circuit component
for the~$\nu^\text{th}$ computational cycle and~$\circ$ denotes composition.
Correspondingly,~$\mathcal{C}$ is represented
by the bit string
\begin{equation}
\bm{B}(\mathcal{C})=\Vert_{\nu=1}^d \bm{B}(\mathcal{C}_\nu),
\end{equation}
where~$\Vert$ denotes concatenation of bit strings.
We also represent each
circuit component~$\mathcal{C}_\nu$
by unitary operator~$G_\nu \in \mathcal{U}(\mathscr{H}_{2^n})$
and represent~$\mathcal C$
by unitary operator~$G\in \mathcal{U}(\mathscr{H}_{2^n})$,
so that
\begin{equation}
G=G_d \, G_{d-1} \cdots  G_1 \, .
\end{equation}

We now focus on RUV-based blind quantum factorization
of odd semiprime~$N$,
wherein factorization circuit~$\mathcal{C}$
acts on two computational registers
of sizes~$t$ and~$L$~\cite{Shors, NielsenChuang}.
Thus, now~$n=t+L$, and the registers are collectively
initialized to~$\ket{0}^{\otimes n}$.
For uncompiled factorization circuits,
$t\geq2\lceil\log{N}\rceil+1$
and~$L\geq\lceil\log{N}\rceil$,
but we employ compilation so these bounds do not apply~\cite{SSV13}.
To convert~$\mathcal{C}$ into
blind circuits~$\mathcal{A}$ \&~$\mathcal{B}$,
we employ a two-step procedure:
first we partition~$\mathcal C$
into a first-stage circuit~$\mathcal{C}_<$
and a second-stage circuit~$\mathcal{C}_>$;
then we convert the sequential computation~$\mathcal{C}_>\circ\, \mathcal{C}_<$
into a bipartite computation~$\mathcal{A}\otimes\mathcal{B}$
on~$\ket\Phi^{\otimes n}$.

In our scheme we require
$\mathcal{A}$ to be a stabilizer circuit,
and~$\mathcal{B}$ to have the minimum possible~$d$.
Thus, before partitioning~$\mathcal{C}$,
we first minimize reduced depth~$d_>$
(defined to be
the number of cycles including
the first non-Clifford cycle and
then all subsequent cycles, whether Clifford or not)
over all circuits that are permutations of the cycles of~$\mathcal{C}$
and are~$\mathcal{C}$-equivalent (i.e., map input to the same output as~$\mathcal{C}$).
In case the minimum reduced depth~$d^*_>$ is not achieved uniquely,
we choose an optimal circuit~$\mathcal{C}^*$.
Then,~$\mathcal{C}_<$ is
the composition of the first~$d-d_>^*$ cycles in~$\mathcal{C}^*$
and~$\mathcal{C}_>$ is
the composition of the last~$d^*_>$ cycles in~$\mathcal{C}^*$.
Thus,
for
\begin{equation}
\mathcal{C}^* = \bigcircle_{\nu=1}^d \mathcal{C}^*_\nu \, ,
\end{equation}
we partition as
\begin{equation}
\label{eq:circuitpartition}
\mathcal{C}_< := \bigcircle_{\nu=1}^{d-d^*_>} \mathcal{C}^*_\nu
\quad  \&\ \quad
\mathcal{C}_> := \bigcircle_{\nu=d-d^*_>+1}^{d} \mathcal{C}^*_\nu  \, ,
\end{equation}
so that
\begin{equation}
\mathcal{C}^* = \mathcal{C}_>\circ\mathcal{C}_< .
\end{equation}
The bit strings~$\bm{B}(\mathcal{C}_<)$
\&~$\bm{B}(\mathcal{C}_>)$
representing~$\mathcal{C}_<$ \&~$\mathcal{C}_>$,
respectively,
are determined by first permuting the component bit strings of~$\bm{B}(\mathcal{C})$
and then partitioning into $\bm{B}(\mathcal{C}_<)\Vert\bm{B}(\mathcal{C}_>)$
following Eq.~\eqref{eq:circuitpartition};
we denote this operation by the bit-string function \textsc{part}.

To establish the second step in our procedure,
we first introduce some notation,
prove a lemma and state a fact.
Below we denote transposition in the computational basis by~${}^\top$.
For~$x=(x_{n}\cdots x_1)\in\{0,1\}^n$, we define
\begin{equation}
\label{eq:Xx}
X^x:=X_{n}^{x_{n}} \otimes \cdots \otimes  X_1^{x_1} \in\mathcal{U}(\mathscr{H}_{2^n}) ,
\end{equation}
where the subscripts below~$X$
indicate the index of the qubit being targeted.
Thus, a computational basis state is
\begin{equation}
\ket{x}=X^x\ket0^{\otimes n}\in \mathscr{H}_{2^n} \, .
\end{equation}
\begin{lemma}
\label{lemma:concat=simult}
For any~$x\in\{0,1\}^n$
and unitary operators~$G_<, G_>, G_A, G_B \in\mathcal{U}(\mathscr{H}_{2^n})$,
the mapping
\begin{equation}
	\label{eq:mapforparallelization}
	X^x G_<^\top \mapsto G_A  \, , \quad X^x G_>\mapsto G_B
\end{equation}
leads to the equality
\begin{equation}
	\label{eq:equationforparallelization}
	X^x G_> G_< X^x\ket0^{\otimes n} = \left( \bra0^{\otimes n} G_A  \otimes G_B\right)\ket\Phi^{\otimes n} .
\end{equation}
\end{lemma}
\begin{proof}
From the ``ricochet'' property~\cite{Wilde17},
\begin{align}
	\label{eq:ricochetproperty}
	&(\ket0^{\otimes n} \bra0^{\otimes n} G_A \otimes G_B) \left(\mathds1\otimes\mathds1 + X\otimes X\right)^{\otimes n} \nonumber\\
	=&\left [ \mathds 1 \otimes \left (G_B G_A^\top \ket0^{\otimes n} \bra0^{\otimes n}  \right) \right] \left(\mathds1\otimes\mathds1 + X\otimes X\right)^{\otimes n} ,
\end{align}
so assign
$G_A\gets X^x G_<^\top$ and
$G_B\gets X^x G_>$.
\end{proof}
\begin{fact}
\label{fact:transpose=reverse}
As each of our primitive gates admits
a symmetric matrix representation in the computational basis,
the operator $G_<^\top$ represents a circuit~$\mathcal{C}_<^\top$
that consists of the components of~$\mathcal{C}_<$ executed in reverse order, i.e.,
\begin{equation}
	\label{eq:transposedcircuit}
	\mathcal{C}_<^\top := \bigcircle_{\nu=d-d^*_>}^{1} \mathcal{C}^*_\nu \, .
\end{equation}
We denote the operation
of obtaining~$\bm{B}(\mathcal{C}_<^\top)$
by reversing the order of components
in~$\bm{B}(\mathcal{C}_<)$
by the bit-string function~\textsc{rev}.
\end{fact}

We now explain how we use Lemma~\ref{lemma:concat=simult} and Fact~\ref{fact:transpose=reverse}
to convert~$\mathcal{C}_>\circ\, \mathcal{C}_<$
into~$\mathcal{A}\otimes\mathcal{B}$.
Let~$G_<, G_>, G_A$ and~$G_B$ be unitary operators
representing~$\mathcal{C}_<, \mathcal{C}_>, \mathcal{A}$ and~$\mathcal{B}$, respectively,
so that Map~\eqref{eq:mapforparallelization} implies
\begin{equation}
\mathcal{A}=X^x \circ \mathcal{C}^\top_< \quad \text{\&} \quad \mathcal{B}= X^x \circ \mathcal{C}_> \, .
\end{equation}
Also, consider any~$x\in\{0,1\}^{t+L}$
with~$x_i=0$ for all~$i\in[t]$.
Then,
\begin{equation}
X^xG_>G_< X^x =G_>G_<=G,
\end{equation}
because the second register of~$\mathcal{C}$
is operated on by only C$^l$NOTs~\cite{Shors},
and Eq.~\eqref{eq:equationforparallelization} simplifies to
\begin{equation}
\label{eq:unitaryparallelization}
G \ket 0^{\otimes n} = \bra{x} G_<^\top \otimes X^x G_> \ket\Phi^{\otimes n}.%
\end{equation}
As the output from only the first register of~$\mathcal{C}^*$
is used to compute~$r$ \cite{Shors},
the~$X^x$ in Eq.~\eqref{eq:unitaryparallelization} can be ignored.
Finally, we have
\begin{equation}
\mathcal{A}=\mathcal{C}_<^\top \quad \text{\&} \quad \mathcal{B}=\mathcal{C}_> \,,
\end{equation}
and B's outcomes~$\{b_i; i\in[t]\}$ are identical
to the output from the first register of~$\mathcal{C}$
whenever A reports~$a_i=0$ for all~$i\in[t]$.
This completes our description of
the procedure to obtain~$\mathcal{A}~\&~\mathcal{B}$
from~$\mathcal{C}$.

\section{Results}
\label{sec:Results}

In this section we present our results, which are two-fold.
Firstly, we present a scalable algorithm to design circuits~$\mathcal{A}~\&~\mathcal{B}$
for blind-quantumly factorizing arbitrary odd semiprime~$N$ using arbitrary base~$a$.
Secondly, we present the outputs of this algorithm,
i.e., the factorization circuits~$\mathcal{A}~\&~\mathcal{B}$,
along with the rest of the blind factorization scheme,
for two cases corresponding to~$N=21$ and~$a=4$.
Our algorithm requires  certain standard functions so we start by specifying a function library,
though we don't reproduce the corresponding algorithms here as they are standard in literature.

We now introduce our concept
for the function library~\textsc{funcLib}
for designing circuits~$\mathcal{A}$
\&~$\mathcal{B}$.~\textsc{funcLib}
comprises four functions,
with two of them~(\textsc{part} \&~\textsc{rev}) already discussed and
the other two well established
in literature on factorization~\cite{BCDP, Shors, SSV13}.
The integer function
\begin{equation}
\label{eq:maximumdepthshorcircuitfunction}
\textsc{maxDep}(N, a) = 96 \lfloor\log{a}\rfloor \lfloor\log{N}\rfloor^2
\end{equation}
yields an upper bound
on factorization-circuit depth, given~$N$ and~$a$,
based on complexity arguments for scaling modular exponentiation~\cite{BCDP}.
The bit-string function~\textsc{ShorCir}
returns a bit string representing the compiled factorization circuit,
given~$N, a, t$ and~$n$~\cite{Shors,SSV13}.
Our procedure for designing
circuits~$\mathcal{A}$ \&~$\mathcal{B}$
is described in Alg.~\ref{algorithm:circuitdesign},
where we employ `type' notation USINT for nonnegative integers and
BIN for bit strings (with~[\,] denoting array size).

\begin{algorithm}[H]
\caption{Parallelizing Factorization}
\label{algorithm:circuitdesign}
\begin{algorithmic}[1]
	\Require
	\Statex USINT \textsc{num}
	\Comment $N=pq$,\; $p\neq q$,\; $p,q\in\mathbb{P}\setminus\{2\}$ \;
	\Statex USINT \textsc{base}
	\Comment Base~$a$: $a<N, \, \gcd(a, N) = 1$ \;
	\Statex USINT \textsc{siz1}
	\Comment First-register size of Shor circuit\;
	\Statex USINT \textsc{siz2}
	\Comment Second-register size of Shor circuit\;
	\Ensure
	\Statex BIN[\;] \textsc{circDesA}, \textsc{circDesB}
	\Procedure{\textsc{cirDesign}}{\textsc{num}, \textsc{base}, \textsc{siz1}, \textsc{siz2}}
	\State Import \textsc{funcLib}
	\Comment For functions\;
	\State USINT \textsc{depth}, \textsc{sizeCom}, \textsc{sizCir}
	\State $\textsc{depth}\gets \textsc{maxDep}(\textsc{num}, \textsc{base})$
	\Comment Eq.~\eqref{eq:maximumdepthshorcircuitfunction}
	\State $\textsc{siz2}\gets\textsc{siz1}+\textsc{siz2}$
	\Comment Replace by total register size
	\State $\textsc{sizeCom} \gets\left\lceil\log  ( \frac{133233}{8!}* \textsc{siz2}**{\frac{\textsc{siz2}}{2}} \,* \textsc{siz2}!)\right\rceil$
	\State $\textsc{sizCir}\gets \textsc{sizeCom} * \textsc{depth}$
	\Comment Space-time product
	\State BIN[\textsc{sizCir}] \textsc{cirDes}, \textsc{cirDesL}, \textsc{cirDesG}, \textsc{cirDesA}, \textsc{cirDesB}
	\State $\textsc{cirDes}\gets$~\Call{\textsc{ShorCir}}{\textsc{num}, \textsc{base}, \textsc{siz1},  \textsc{siz2}}
	\State $\textsc{cirDesL}\|\textsc{cirDesG} \gets \textsc{part}(\textsc{cirDes})$
	\Comment Optimal partition
	\State $\textsc{cirDesA}\gets\textsc{rev}(\textsc{cirDesL})$
	\Comment Reverse order
	\State $\textsc{cirDesB}\gets\textsc{cirDesG}$
	\EndProcedure
\end{algorithmic}
\end{algorithm}

Next we describe the computational subprotocol for
our scheme.
Prior to executing the subprotocol,
C runs Alg.~\ref{algorithm:circuitdesign}
to design~$\mathcal{A}$ \&\ $\mathcal{B}$
and sends the output bit strings to the servers.
Then, in every instance of the subprotocol,
C instructs A \&\ B to execute their circuits
and report measurement outcomes.
If~$a_i=0$ for all~$i\in[t]$,
which occurs with probability~$2^{-t}$,
C computes a candidate for~$r$ by
classically processing~$\{b_i; i\in[t]\}$.
For completeness, we describe this classical processing
in Appendix~C.
The exponentially small probability
could be improved by instructing B to
Pauli-correct before computing,
but doing so blindly would entail higher space requirements~\cite{RUV, Gottesman3};
in this work we instead focus on~$N$
with~$t$ sufficiently small for~$2^{-t}$ to be a feasible probability.

This concludes our discussion of the computation subprotocol,
whereas the three other RUV subprotocols
are standard so we describe them
in Appendix D.
In our full multi-round protocol
for blind quantum factorization
(summarized in Appendix D),
C runs one of the four sub-protocols at random,
interacting with the servers as required,
until the correct~$r$ is found.
Although tomographic verification via
the stabilizer framework~\cite{RUV, HZM+17}
is inapplicable here
due to our inclusion of
non-Clifford operators~\cite{RUVProof, JKMW01},
C can classically verify~$r$
in~$\operatorname{polylog}N$ time
by checking~$\gcd(p,N)=p\in\mathbb{P}$~\cite{GKK19}.
\begin{figure}
\begin{tikzpicture}
	\node[anchor=south west,inner sep=0] (image) at (0,0) {\includegraphics[width=0.8 \columnwidth]{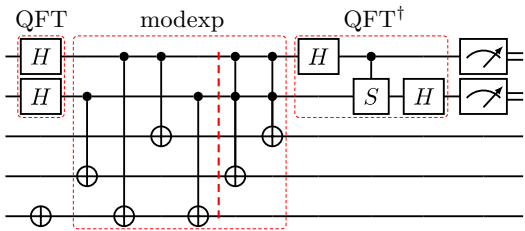}};
	\begin{scope}[
		x={(image.south east)},
		y={(image.north west)}
		]
		\node [black] at (0.71, 1.08) {QFT$^{\dagger}$};
		\node [black] at (0.34, 1.07) {modexp};
		\node [black] at (0.07, 1.07) {QFT};
		\draw[red,thick, densely dashed] (0.41,0.06) -- (0.41,0.92);
	\end{scope}%
\end{tikzpicture}%
\caption{%
	Optimally permuted~$\mathcal{C}^*$ for~$N=21,a=4$, $t=2$ and $L=3$.
	Unlabeled input states are~$\ket0$
	and measurements are in the $Z$ basis.
	Quantum Fourier transform (QFT) is performed using Hadamard~($H$) gates,
	modexp using CNOT and C$^2$NOT gates,
	and inverse QFT~($\text{QFT}^\dagger$) using a controlled-phase gate~($S$) and Hadamard gates.
	The vertical dashed line (red) indicates the optimal partition of the circuit.}%
\label{fig:factorization21with5qubits}
\end{figure}
\begin{figure}
\begin{tikzpicture}%
	\node[anchor=south west,inner sep=0] (image) at (0,0) {\includegraphics[width=0.98\columnwidth]{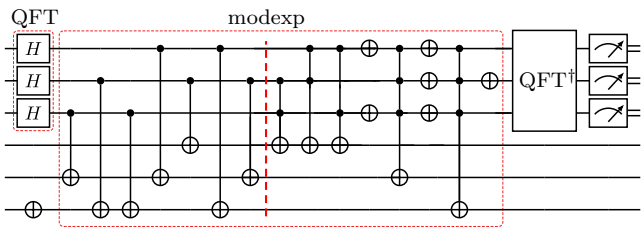}};
	\begin{scope}[
		x={(image.south east)},
		y={(image.north west)}
		]
		\node [black] at (0.048, 1.05) {\footnotesize QFT};
		\node [black] at (0.411, 1.05) {\footnotesize modexp};
		\node [black] at (0.852, 0.748) {\footnotesize QFT$^{\dagger}$};
		\draw[red,thick, densely dashed] (0.410,0.067) -- (0.410,0.94);
	\end{scope}%
\end{tikzpicture}%
\caption{%
	Optimally permuted~$\mathcal{C}^*$ for~$N=21,a=4$, $t=3$ and $L=3$.
	Unlabeled input states are~$\ket0$
	and measurements are in the $Z$ basis.
	QFT is performed using Hadamard~($H$) gates,
	modexp using NOT, CNOT, C$^2$NOT, and C$^3$NOT gates.
	The inverse QFT (QFT$^{\dagger}$) is abridged as a 3-qubit gate for readability.
	The vertical dashed line (red) indicates the optimal partition of the circuit.}%
\label{fig:factorization21with6qubits}
\end{figure}

We now present circuits~$\mathcal{C}^*$,~$\mathcal{A}$ and~$\mathcal{B}$
specifically for~$N=21$ and~$a=4$
with two cases considered,
namely~$t=2$ and~$t=3$,
both with~$L=3$.
The five-qubit~($t=2$) and six-qubit~($t=3$)
optimally permuted circuits~$\mathcal{C}^*$
are shown in Figs.~\ref{fig:factorization21with5qubits}
and~\ref{fig:factorization21with6qubits},
respectively,
where a vertical dashed line indicates the partition
into a first
and a second stage.
For both cases,
the output from the first register
of~$\mathcal{C}^*$
is shown in Fig.~\ref{fig:outputsall} in Appendix C.
For~$t=2$,
the blind circuits~$\mathcal{A}$ \&~$\mathcal{B}$
are shown in Fig.~\ref{fig:blindfact21}.
For~$t=3$,
we do not show~$\mathcal{A}$ \&~$\mathcal{B}$ explicitly
as they could be similarly obtained
from~$\mathcal{C}^*$ in Fig.~\ref{fig:factorization21with6qubits}:
$\mathcal{A}$ is the first stage performed in reverse order
whereas~$\mathcal{B}$ is the second stage.
Both~$\mathcal{A}$ \&~$\mathcal{B}$ terminate with~$Z$ measurements of all qubits.

\section{Discussion}
\label{sec:Disc}

In this section, we discuss the significance
some implications of our results.
Our Alg.~\ref{algorithm:circuitdesign} designs
scalable blind quantum factorization circuits
for arbitrary~$N$ and~$a$,
with expected number of runs growing as~$2^t$
(note that we use compiled circuits,
so~$t$ is not restricted by the~$2 \lceil \log N \rceil+1$ bound).
Alg.~\ref{algorithm:circuitdesign}, via \textsc{Part}, also ensures
that~$\mathcal{A}$ is always a stabilizer circuit,
and thus not resource-intensive to implement.
Moreover,~$\mathcal{A}$ can be implemented fault-tolerantly straightforwardly~\cite{Gottesman3}
and affords complete tomographic verification via just CHSH measurements~\cite{RUV}.

On the other hand, we expect~$\mathcal{B}$ to include non-Clifford operations
and be resource intensve (potentially experimentally infeasible at present).
This can also be seen in the circuits~$\mathcal{C}^*$,~$\mathcal{A}$ and~$\mathcal{B}$
for~$t\in\{2,3\}$.
For~$t=2$,~$\mathcal{C}^*$ has~$d=8$,~$d^*_>=5$
and incorporates three non-Clifford gates
(one~CR$^1$ and two~C$^2$NOTs);
$\mathcal{A}$ \&~$\mathcal{B}$ thus
have depths of three \& five, respectively.
For~$t=3$,~$\mathcal{C}^*$
has~$d=17$,~$d^*_>=13$ and incorporates eight non-Clifford gates
(one~CR$^2$, two~CR$^1$s,
three~C$^2$NOTs and two~C$^3$NOTs);
$\mathcal{A}$ \&~$\mathcal{B}$ thus
have depths of four \& thirteen, respectively.
In both cases,~$\mathcal{A}$
is conveniently a stabilizer circuit
whereas~$\mathcal{B}$ incorporates
all non-Clifford gates in~$\mathcal{C}^*$.

For~$t=2$,
despite following Shor's algorithm~\cite{Martin},
$\mathcal{C}^*$ never yields the correct period
due to insufficient space in the first register;
hence, $\mathcal{A}$ \&~$\mathcal{B}$
fail to factorize~21.
However, the output from~$\mathcal{C}^*$
is sufficient to establish a proof-of-concept as in Ref.~\cite{Martin}.
Further, photonic implementations of~$\mathcal{A}$ \&~$\mathcal{B}$,
which entail scaling up from
one C$^2$NOT to two~\cite{Lanyon}
and from three EPR pairs to five~\cite{HZM+17},
are more feasible
in this case compared to~$t=3$.
For~$t=3$, $\mathcal{C}^*$
delivers the correct period
with probability~0.47
so~$\mathcal{A}$ \&~$\mathcal{B}$
succeed in factorizing with probability~0.058,
but~$\mathcal{B}$ requires significant resources.

Another advantage of our blind quantum factorization scheme
is the low qubit count required to factorize~$N$.
Our scheme requires, at worst,~$\mathrm{O}(\log N)$ qubits to guarantee factorization
compared to the~$\mathrm{O}((\log N)^2)$ qubits
a literal adaptation of the RUV protocol to quantum factorization would require.
Finally, we remedy RUV's open problem of tomographic verifiability
of non-Clifford computations in the context of factorization
by declaring B to be dishonest if C finds A honest
(which requires only CHSH measurments)
but still does not recieve the correct period
(which can be checked efficiently classically)
from B.

\begin{figure}
\begin{tikzpicture}
	\node[anchor=south west,inner sep=0] (image) at (0,0) {\includegraphics[width=0.8\columnwidth]{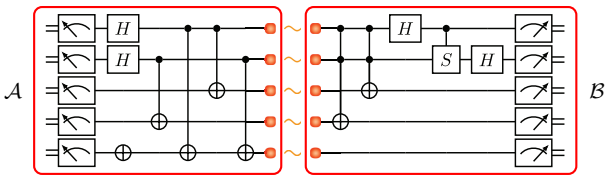}};
	\begin{scope}[%
		x={(image.south east)},
		y={(image.north west)}
		]
		\node [black] at (-0.06, 0.5) {\footnotesize $\mathcal{A}$};
		\node [black] at (1.06, 0.5) {\footnotesize $\mathcal{B}$};
		\draw [thick,red,rounded corners=3] (-0.02, -0.04) rectangle (0.454, 1.04);
		\draw [thick,red,rounded corners=3] (0.501, -0.04) rectangle (1.02, 1.04);
	\end{scope}%
\end{tikzpicture}%
\caption{%
	Blind circuits~$\mathcal{A}$ \&~$\mathcal{B}$
	for $N=21, a=4$, $t=2$ and $L=3$.
	$\mathcal{A}$ \&~$\mathcal{B}$ (rounded rectangles)
	each act on input registers initialized to one half of the bipartite state~$\ket\Phi^{\otimes 5}$ (red dots).
	$\mathcal{A}$ comprises NOT, CNOT, and Hadamard~($H$) gates
	whereas~$\mathcal{B}$ comprises C$^2$NOT, controlled-phase~($S$) and Hadamard gates.
	All measurements are in the~$Z$ basis.}
\label{fig:blindfact21}
\end{figure}

\section{Conclusions}
\label{sec:Conc}

Here
we have developed a BQC scheme
for factorizing~21.
Our multi-round protocol is accessible
to a classical client,
who communicates with
two remote quantum servers.
The servers send
the client~$Z$-measurement results
for each round.
By processing these data,
the client determines candidates
for factors of~21
or verifies honesty of the servers,
all while concealing the actual task.
Our choice of hard~$a$ implies that
servers employ non-Clifford gates,
which is a non-trivial requirement
unseen for $N=15$~\cite{HZM+17};
our non-Clifford analysis
establishes a foundation for
future BQC factorization protocols.
Finally, our protocol for~$t=2$ motivates
a challenging but feasible photonic experiment
that would set a milestone towards
secure quantum computation
for classical clients.

\begin{acknowledgments}
Aritra Das thanks and acknowledges financial support from the Shastri Indo-Canadian Institute through their Shastri Research Student Fellowship program
and from the Australian Research Council Centre of Excellence CE170100012.
Aritra Das is also grateful to the Indian Institute of Technology Kanpur,
where a majority of this work was undertaken.
Finally, Aritra Das and Barry C. Sanders acknowledge the traditional owners of the land on which this work was undertaken at the University of Calgary: the Treaty 7 First Nations.
\end{acknowledgments}

\bibliographystyle{apsrev4-2}

\section*{Appendix A: Decomposition of Primitives into Standard Set}
\label{sec:AppA}
Our primitive gates are composites of
the usual primitive set of CNOT, $H$
and the R$^2$ gate called~$T$,
and we explain briefly how
our gates decompose to this ``standard set''~\cite{BMPRV00}.
A CR$^k$ gate
can be performed using two CNOTs
and two~R$^{k+1}$ gates~\cite{BBCD+95}.
A C$^{l}$NOT (for $l\geq3$) can be
performed using~$l-2$ ``dirty'' ancill\ae\
and~$4l-8$ C$^2$NOTs~\cite{BBCD+95},
or~$l-2$ ``clean'' ancill\ae\
and~$2l-3$ C$^2$NOTs~\cite{MDM05}.
Here, ``clean'' ancill\ae\ are qubits initialised to some computational basis state~$\ket x$,
whereas
``dirty'' ancill\ae\ are in some unknown state in~$\mathscr{H}_2$
to which they must be restored post-computation.
Finally, each C$^2$NOT requires
seven~$T$ gates, two~$H$ gates, one R$^1$ gate and six CNOTs~\cite{GKMR14, SM09}.

\section*{Appendix B: Size of Bit-string Representation}
\label{sec:AppB}
Now we establish an upper bound on
the size of a bit string~$\bm{B}(\mathcal{C}_\nu)$
representing any circuit component~$\mathcal{C}_\nu$
over~$n$ qubits.
An upper bound on this size is given by~$\lceil\log (s(n))\rceil$,
for~$s(n)$ the number of distinct circuit components
(including the identity component~$\mathds{1}^{\otimes n}$)
on~$n$ qubits.
Correspondingly,
we posit an upper bound on~$s(n)$,
from which the bound on size follows.

First consider the closely related quantity~$s'(n)$,
which is the number of distinct circuit components over~$n$ qubits,
ignoring the argument~$k$ of the CR gates.
Then, for~$n\geq 1$,
we have the recurrence relation
\begin{align}
\label{eq:exactrecurrence}
s'(n) = &\; 3 s'(n-1) + 2 (n-1) s'(n-2)  \nonumber\\
&+ \sum_{j=0}^{n-2} \binom{n-1}{j} (n-j) s'(j).
\end{align}
Along with initial conditions $s'(0)=1$ and $s'(n<0) =0$,
Eq.~\eqref{eq:exactrecurrence} can be used to compute~$s'(n)$
for any~$n\geq1$.
Next we prove the bound
\begin{equation}
\label{eq:upperboundnumberofcycles}
s'(n) < c \, n! \quad \forall\,\,  n\in\mathbb{N} ,
\end{equation}
where~$c$ is some positive constant.
Numerical evidence suggests~$s'(n)/n!$ increases up to~$n=8$,
and then decreases monotonically.
Correspondingly, we choose
\begin{equation}
\label{eq:cvalue}
c := \max_{1\leq n\leq8} \frac{s'(n)}{n!} = \frac{1133233}{8!}
\end{equation}
so that inequality~\eqref{eq:upperboundnumberofcycles}
is satisfied for~$n\in[8]$ trivially.

Next we show that inequality~\eqref{eq:upperboundnumberofcycles}
holds for~$n>8$
by using strong induction on~$n$
with base case~$n=8$.
In the inductive step,
we prove that 
inequality~\eqref{eq:upperboundnumberofcycles}
holds for~$n>8$
if it holds for all~$j\in[n-1]$.
From Eq.~\eqref{eq:exactrecurrence} we have
\begin{equation}
s'(n) 	\leq  c \; \left [ 7 (n-1)! + \sum_{j=0}^{n-3} \, \binom{n-1}{j}  (n-j) j! \right ] \, .
\end{equation}
Using $\sum_{j=0}^{n-1} \nicefrac{1}{j!} < e$,
we get
\begin{equation}
\sum_{j=0}^{n-1} \, \binom{n-1}{j} (n-j) j! < 2 e (n-1)!
\end{equation}
so that
\begin{equation}
s'(n)  < c (4+2e)(n-1)! < c n! \, ,
\end{equation}
where this last inequality relies on~$n>8$.

Note that
there are at most~$\lfloor \nicefrac{n}{2}\rfloor$ CR gates
in any~$n$-qubit component
and~$n$ choices of~$k$
($0\leq k<n$)
for each CR gate.
Thus,
\begin{equation}
s(n) \leq  n^{\nicefrac{n}{2}} s'(n) \leq c \, n^{\nicefrac{n}{2}} \, n! \,
\end{equation}
and the size of~$\bm{B}(\mathcal{C}_\nu)$ is at most~$\lceil\log (c \, n^{\nicefrac{n}{2}} \, n!)\rceil$.

\section*{Appendix C: Recovering Factors from Circuit Output}
\label{sec:AppC}
\begin{figure}%
\begin{tikzpicture}%
	\node[anchor=south west,inner sep=0] (image) at (0,0) {\includegraphics[width=180pt]{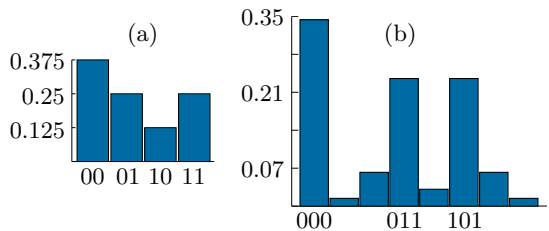}};
	\begin{scope}[%
		x={(image.south east)},
		y={(image.north west)}%
		]%
		\clip (-0.13,-0.13) rectangle (1,1.05);
		\node at (0.69, 0.9) {(b)};
		\node at (0.15, 0.9) {(a)};
		\node  at (0.045,0.16) {$00$};
		\node at (0.12,0.16) {$01$};
		\node at (0.187,0.16) {$10$};
		\node at (0.257,0.16) {$11$};
		\node at (-0.07, 0.4) {$0.125$};
		\node at (-0.058, 0.58) {$0.25$};
		\node at (-0.07, 0.76) {$0.375$};
		\node at (0.510,-0.07) {$000$};
		\node at (0.7,-0.07) {$011$};
		\node at (0.826,-0.07) {$101$};
		\node at (0.402, 0.2) {$0.07$};
		\node at (0.402, 0.596) {$0.21$};
		\node at (0.402, 0.99) {$0.35$};
	\end{scope}%
\end{tikzpicture}%
\caption{%
	Mean probabilities of measurement
	of the first register of $\mathcal{C}^*$
	for (a) $t=2$ and (b) $t=3$.
	The rightmost bit corresponds to the topmost qubit
	in Figs. \ref{fig:factorization21with5qubits}
	\& \ref{fig:factorization21with6qubits}.}
\label{fig:outputsall}
\end{figure}
Here we briefly describe the classical post-processing procedure
to obtain a candidate~$r$,
and thereby factors~$p$ and $q$,
from the first-register output of a factorization circuit~\cite{Shors}.
Given an output~$y\in\{0,1\}^t$ from the first register,
the procedure involves
calculating continued fraction convergents~$\nicefrac{d}{s}$
for~$\nicefrac{y}{2^t}$
such that~$s<N$ and
$\vert\nicefrac{d}{s}-\nicefrac{y}{2^t}\vert<\nicefrac{1}{2^{t+1}}$.
Then,~$s$ is a candidate for~$r$ and
if $a^s=1\pmod{N}$ and if $a^{\nicefrac{s}{2}} \neq -1 \pmod{N}$,~$s$
is the period~$r$.
Finally, the factors of~$N$ are calculated as $p,q=\gcd(a^{\nicefrac{r}{2}}\pm1, N)$.

For~$N=21$,
we show the output from the first register of~$\mathcal{C}^*$
along with their measurement probabilities
in Fig.~\ref{fig:outputsall}(a) and (b)
for the~$t=2$
and the~$t=3$ circuit, respectively.
Post-processing each measurement outcome for~$t=2$,
we find that the~$t=2$ circuit never yields the correct~$r$.
However, the correct period could be obtained via additional heuristic post-processing,
which involves checking whether multiples of~$s$ or~$s\pm1, s\pm2, \dots$ are the period.
Moreover, the measurement outcomes in Fig.~\ref{fig:outputsall}(a) have been experimentally verified~\cite{Martin}.
In contrast, the~$t=3$ circuit yields the correct~$r$ for~$y=011$ and~$y=101$.
Each of these two outcomes occurs with a probability of~0.235,
so the~$t=3$ circuit succeeds in factorizing~21 with probability~0.47.

\section*{Appendix D: Summary of Overall Protocol}
\label{sec:AppD}
In this section, we provide a summary of our protocol
for blind quantum factorization,
which comprises four subprotocols.
We have already discussed the computational subprotocol
in the main text,
whereas our implementation of the three other
subprotocols---the CHSH
and both tomography subprotocols---is standard~\cite{RUV, HZM+17}.
Thus, we first briefly describe these three subprotocols,
and then we outline the overall protocol for C.

In the CHSH subprotocol, C runs multiple rounds of the CHSH game
between the servers~\cite{CHSH, HZM+17}.
CHSH rigidity ensures that,
if the servers win an optimal fraction of rounds,
their shared resource is indeed~$\ket\Phi$ and
they have been honest in computational basis measurements~\cite{RUV, RUVProof}.
The two tomography subprotocols
verify whether one server has computed correctly
by collating the other server's
simultaneous~$X$- or~$Z$-basis measurement outcomes.
Whereas such measurements are sufficient
for state verification in the case
$N=15$
(due to only stabilizer circuits being used),
they are insufficient for exact tomography for general~$N$~\cite{RUVProof, JKMW01}.
Regardless, C can classically verify~$r$
in $\operatorname{polylog}N$ time
by checking~$\gcd(p,N)=p\in\mathbb{P}$~\cite{GKK19}.

The full multi-round protocol to factorize~$N$ blindly
is as follows.
C runs consecutive rounds
until either the correct period is obtained
or any server dishonesty is flagged.
In each round,
she randomly executes one of the four RUV subprotocols.
Specifically,
she executes the computational subprotocol
with some small probability~$\eta$
and the three security subprotocols
(one CHSH and two tomography subprotocols)
with probability~$\nicefrac{(1-\eta)}3$ each.
The optimal choice of~$\eta$
is the sweet spot of factorizing successfully
traded against detecting server malfeasance
and, in practice,
would be obtained by trial and error.
\end{document}